\documentclass[a4paper]{amsart}
\usepackage[foot]{amsaddr}
\usepackage[T1]{fontenc}
\usepackage{fullpage}
\usepackage{amsmath,amsthm,amssymb,amscd,amsfonts}
\usepackage{mathtools}
\usepackage[dvipsnames]{xcolor}
\definecolor{darkgreen}{rgb}{0,0.5,0}
\definecolor{darkblue}{rgb}{0,0,0.5}
\usepackage[colorlinks=true,urlcolor=darkblue,citecolor=darkgreen,linkcolor=darkblue]{hyperref}
\usepackage{graphicx}
\usepackage{lmodern}
\PassOptionsToPackage{hyphens}{url}

\newcommand{\bbF}{\mathbb{F}}
\newcommand{\GLn}{\mathrm{GL}_n(\bbF_2)}
\newcommand{\additive}{\mathsf{L}}
\newcommand{\reduced}{\mathsf{L}^*}
\newcommand{\ind}{\mathrm{ind}}
\newcommand{\Ut}{U^{\!\top}}
\usepackage[sc]{mathpazo}
\newtheorem{thm}{Theorem}[section]
\newtheorem{prop}[thm]{Proposition}
\newtheorem{lem}[thm]{Lemma}

\theoremstyle{remark}
\newtheorem{remx}[thm]{Remark}

\theoremstyle{definition}
\newtheorem{defn}[thm]{Definition}

\newcommand{\qedblack}{\hfill \scalebox{0.7}{\ensuremath{\blacksquare}}}
\newenvironment{rem}
  {\pushQED{\qedblack}\begin{remx}}
  {\popQED\end{remx}}
\begin{document}

\title{Lower bounds for the CNOT-complexity of linear reversible operators}

\author{S\o ren Fuglede J\o rgensen}
  \address{Kvantify, Rosenvængets Allé 25, DK-2100 Copenhagen, Denmark}
  \email{sfj@kvantify.dk}

\date{\today}

\begin{abstract}
    The CNOT-complexity of an invertible matrix over $\bbF_2$ is the minimum number of CNOT gates needed to synthesize the corresponding linear reversible operator. While the maximum CNOT-complexity over all $n \times n$ matrices is known to be $\Theta(n^2 / \log n)$, no explicit family of matrices requiring a superlinear number of CNOT gates is known, and until now the hardest explicitly known family has been the cyclic permutations, with CNOT-complexity $3(n-1)$. We show that lower bounds for the additive complexity of not-necessarily-reversible linear operators can be lifted to the reversible setting with only a small loss. As an application, we use this to describe an explicit family of matrices, constructed from parity-check matrices of error-correcting codes, with CNOT-complexity at least $4n - o(n)$, asymptotically surpassing the cyclic permutations. Moreover, this construction yields an explicit matrix $A \in \mathrm{GL}_{n}(\bbF_2)$, $n = 17167$, whose CNOT-complexity exceeds that of the cyclic permutation on $n$ symbols.
\end{abstract}

\maketitle


\section{Introduction}
Let $\bbF_2 = \{0, 1\}$ be the field with two elements and let $\GLn$ be the group of invertible $n \times n$ matrices over $\bbF_2$. For $i, j \in \{1, \dots, n\}$ with $i \neq j$, let $E_{i, j} \in \bbF_2^{n \times n}$ be the matrix with a single $1$ in the $(i, j)$-th entry, and let $T_{i,j} = I + E_{i, j} \in \GLn$ be the corresponding elementary matrix, also referred to as a transvection.

Any matrix $A \in \GLn$ can be expressed as a product of transvections, for example by Gaussian elimination, and the \emph{CNOT-complexity} $d(A)$ of $A$ is defined to be the minimum number of transvections in any such product.

In reversible and quantum computing, each transvection $T_{i, j}$ corresponds to a CNOT gate with control qubit $j$ and target qubit $i$, and the CNOT-complexity of $A$ is the minimum number of CNOT gates needed to implement the linear reversible operator corresponding to $A$. A collection of CNOT gates is also called a linear reversible circuit, and the process of finding a product of transvections equal to $A$ is referred to as \emph{synthesizing} a linear reversible circuit for $A$. Minimizing the number of CNOT gates that implement a given linear reversible operator is an important practical problem. For example, peephole optimization of more general quantum circuits consisting of CNOT and single-qubit gates involves identifying and minimizing CNOT-only subcircuits. Reducing these subcircuits decreases the overall number of two-qubit gates and typically reduces the circuit's overall error rate. Accordingly, considerable research has been devoted to algorithms and heuristics for minimizing CNOT gate counts in linear reversible circuits \cite{Kazuo2002, Patel2008, Amy2013, Meuli2018, Amy2018, Brugiere2020, Brugiere2021, Bataille2022, Schneider2022, Gheorghiu2023, Shaik2024, Bu2025, Christensen2025, Webster2025}.

The present work concerns lower bounds on the CNOT-complexity of families of linear reversible operators, particularly the problem of finding explicit families with high CNOT-complexity. Patel, Markov, and Hayes \cite{Patel2008} established that
\[
    \max_{A \in \GLn} d(A) = \Theta(n^2 / \log n),
\]
but did not provide an explicit family of matrices with CNOT-complexity $\Theta(n^2 / \log n)$. Indeed, finding any family requiring a superlinear number of CNOT gates is an open problem believed to be difficult \cite{Gowers2011}.

To our knowledge, the hardest explicitly known family $P_n$ consists of the permutation matrices corresponding to the cyclic permutations $(1, 2, \dots, n)$. Such a permutation is the product of $n - 1$ transpositions, each of which can be expressed as a product of three transvections, and thus $d(P_n) \leq 3(n-1)$. It is natural to ask whether this example is extremal for each fixed $n$, and hence whether $3(n-1)$ is an upper bound on the CNOT-complexity of all elements of $\GLn$. Recent work of Bu, Fan, and Joo \cite{Bu2025} shows that $d(P_n) = 3(n-1)$, and results in \cite{Bataille2022, Christensen2025} show that for $n \leq 7$, the maximum CNOT-complexity is attained by a permutation matrix, but that there exist an $n_0$ with $8 \leq n_0 \leq 20$ and a matrix $A_{n_0} \in \mathrm{GL}_{n_0}(\bbF_2)$ such that $d(A_{n_0}) > 3(n_0-1)$. In fact, a counting argument shows that, asymptotically, most matrices in $\GLn$ have CNOT-complexity greater than $3(n-1)$, yet no explicit family has been described.

In this work, we describe an explicit family of matrices $A_n$, constructed from the parity-check matrices of error-correcting codes, that satisfy $d(A_n) \geq 4n - o(n)$, asymptotically surpassing the CNOT-complexity of the cyclic permutations. Moreover, this construction gives rise to an explicit matrix $A \in \mathrm{GL}_{17167}(\bbF_2)$ with $d(A) > 3 \cdot (17167 - 1)$.

Our results build on recent work \cite{Sergeev2025} on the related problem of establishing lower bounds on the additive complexity of linear operators that are not necessarily reversible. We show that such bounds can be lifted to the reversible setting with only a small loss. Since proving lower bounds may be easier in the non-reversible setting, this yields a new approach to proving lower bounds on the CNOT-complexity of linear reversible operators. Because the concrete example relies only on general properties of error-correcting codes, the results may also be of practical interest and may provide tools for establishing lower bounds on the gate counts of encoding and syndrome-extraction circuits.

\subsection*{Acknowledgements}
This work was funded by the Innovation Fund Denmark (Grand Solutions) -- grant no. 5366-00005B.

\section{Background on additive complexity}
As noted in the introduction, we will be relying on results from the literature on additive complexity of linear operators that are not necessarily reversible. These are implemented by circuits that still consist of gates that compute the sum (XOR) in $\bbF_2$ of two input bits. Unlike the reversible setting, the total number of input and output bits may differ as output bits may be discarded, and we allow arbitrary fan-out and swapping of bits, as well as auxiliary ``scratch'' values. Let $\oplus$ denote addition in $\bbF_2$.
\begin{defn}
    An \emph{additive circuit} on inputs $x_1, \dots, x_n$ is a finite sequence of \emph{gates} $g_1, g_2, \dots, g_s$; each gate is an ordered pair of earlier values, $g_k = u \oplus v$, where $u, v \in \{x_1, \dots, x_n, g_1, \dots, g_{k-1}\}$, computing their sum in $\bbF_2$. When $n=0$, we regard the unique linear form on $\bbF_2^0$, namely zero, as initially available.

    The circuit \emph{computes} a matrix $A \in \bbF_2^{m \times n}$ if among $x_1, \dots, x_n, g_1, \dots, g_s$ one can designate $m$ values equal to $(Ax)_1, \dots, (Ax)_m$ for any input $x = (x_1, \dots, x_n) \in \bbF_2^n$. Its \emph{size} is the number $s$ of gates, and the \emph{additive complexity} of $A$ is defined as the smallest size of any additive circuit computing $A$,
    \[
        \additive(A) = \min\{s \mid \text{there exists an additive circuit of size $s$ computing $A$}\}.
    \]
\end{defn}
As is the case \cite{Patel2008} for CNOT-complexity, the maximum additive complexity of $n \times n$ matrices is known to be $\Theta(n^2 / \log n)$, yet no explicit family with a superlinear lower bound is known \cite{Jukna2013}.
Indeed, the recent work of Sergeev \cite{Sergeev2025} improves the best known lower bound for additive complexity of an explicit family of matrices from $3n - o(n)$ to $5n - o(n)$.

Now, any linear reversible circuit also defines an additive circuit computing the same matrix, using the same number of gates: each CNOT gate creates the sum of the two current wire values and updates the target's designated value. In particular, this holds for an optimal linear reversible circuit, which proves the following:

\begin{prop}
  For any $A \in \GLn$, we have $\additive(A) \leq d(A)$.
\end{prop}

\section{Lower bounds for linear reversible operators}
Our aim is to extend the lower bounds for additive complexity to the reversible setting. To do so, we first consider the effect of padding matrices to make them square and invertible, and later we will consider modifications to the procedure that allow us to lift lower bounds on complexity.
\begin{lem}
    \label{lem:submatrix}
    If $A$ is a submatrix of $M \in \bbF_2^{m \times n}$, then $\additive(A) \leq \additive(M)+1$.
\end{lem}
\begin{proof}
    Let $I$ and $J$ be the sets of indices of the rows and columns of $M$ retained in $A$, and let $C$ be an optimal additive circuit computing $M$. If $J=\varnothing$, then $A$ has no columns, so each of its rows is the unique linear form on $\bbF_2^0$. By the convention in the definition, $\additive(A)=0$, and the result follows.

    Suppose now that $J\neq\varnothing$, and write $y_1,\dots,y_{\lvert J\rvert}$ for the inputs of a circuit for $A$, indexed in the same order as the elements of $J$; denote the position of $j\in J$ in this order by $\ell(j)$. First use one gate to compute $z=y_1\oplus y_1=0$. Replay the gates of $C$ in their original order, replacing each retained input $x_j$, $j\in J$, by $y_{\ell(j)}$, and each deleted input $x_j$, $j\notin J$, by $z$. Inductively, every replayed gate computes the restriction of the corresponding value of $C$ obtained by setting the deleted inputs to zero. In particular, for every $i\in I$, the value corresponding to the $i$-th designated output of $C$ is
    \[
      \sum_{j\in J}M_{i,j}y_{\ell(j)}=(Ay)_i.
    \]
    Designating these values as the outputs gives a circuit for $A$ with at most $\additive(M)+1$ gates.
\end{proof}
Now, any square matrix can be embedded as a submatrix of an invertible matrix by extending its row and column spaces to a basis.
\begin{lem}
  \label{lem:embedding}
  Let $A \in \bbF_2^{n \times n}$ have corank $s = n - \mathrm{rank}(A)$. Then $A$ is the top-left $n \times n$ block of an invertible matrix $M \in \mathrm{GL}_N(\bbF_2)$ with $N = n + s$.
\end{lem}
Now suppose that $\{A_{n}\}$ is a family of matrices with $\additive(A_{n}) \geq cn - o(n)$ for some $c > 0$ and corank $s = \gamma_n n$. Then by the above lemmas, we can construct a family of invertible matrices $\{M_{N}\}_{N \in I}$ such that for every $N \in I$,
\begin{align}
  \label{eq:embedding-bound}
  d(M_{N}) \geq \additive(A_{n}) - 1 \geq cn - o(n) = \frac{c}{1 + \gamma_n} N - o(N).
\end{align}
As described in the previous section, \cite{Sergeev2025} gives a family of matrices $A_n$ with $c = 5$. However, without also bounding the corank, the above procedure would not suffice to give a family whose CNOT-complexity exceeds that of the cyclic permutations.
\begin{lem}
  \label{lem:add-permutation}
  For $A, B \in \bbF_2^{n \times n}$, we have
  \[
    \additive(A + B) \leq \additive(A) + \additive(B) + n.
  \]
  Moreover, any permutation matrix $P$ has $\additive(P) = 0$ and hence $\additive(A + P) \geq \additive(A) - n$.
\end{lem}
\begin{proof}
    Let $C_A$ and $C_B$ be optimal additive circuits for $A$ and $B$ respectively. Then, use $n$ additional gates to combine their designated outputs coordinate-wise. The resulting circuit computes $A + B$.

    Let $P$ be a permutation matrix. Then any gate-free circuit computes $P$ by designating the outputs according to the permutation, and from the first part of the lemma,
    \[
      \additive(A) = \additive((A + P) + P) \leq \additive(A + P) + n. \qedhere
    \]
\end{proof}
In other words, adding a permutation matrix to a matrix $A$ does not change its additive complexity by more than $n$, yet it can change its rank by as much as $n$. In particular, the following permutation matrix will be useful:
\begin{lem}
  \label{lem:permutation-matrix}
  For $U \in \bbF_2^{a \times b}$, write $A_U=\begin{psmallmatrix}U&0\\0&\Ut\end{psmallmatrix} \in \bbF_2^{n \times n}$, $n = a +b$. Let
  \[
    P_0 = \begin{pmatrix} 0 & I_a \\ I_b & 0 \end{pmatrix}, \quad M_0 = A_U + P_0 = \begin{pmatrix} U & I_a \\ I_b & \Ut \end{pmatrix}.
  \]
  Then $\ker M_0 \cong \ker (I_b + \Ut U)$ and $\mathrm{corank}(M_0) \leq b$.
\end{lem}
\begin{proof}
  Projection onto the first $b$ coordinates gives an isomorphism \mbox{$\ker M_0 \to \ker(I_b + \Ut U)$}, with inverse $x \mapsto (x, U x)$, and
  \[
    \mathrm{corank}(M_0) = \dim \ker M_0 = \dim \ker(I_b + \Ut U) \leq b. \qedhere
  \]
\end{proof}

\section{Asymptotically hard instances}
We are now in a position to construct the explicit hard instances as described in the introduction. Here, we consider a matrix $A$ to be ``hard'' if it has CNOT-complexity $d(A) > 3(n-1)$. Recall also that asymptotically, most matrices in $\GLn$ are hard in this sense, yet no explicit families have been described. Here, to be precise, we consider a family $\{A_n\}_{n \in I}$, $I \subseteq \mathbb{N}$, of matrices to be explicit if there is a deterministic algorithm that for given $n \in I$ and a pair of indices $(i, j)$, outputs $(A_n)_{i, j}$ in time polynomial in $n$.

In \cite{Sergeev2025}, it is noted that it is possible to bound the additive complexity of a matrix in terms of the independence index of its rows. Here, for an indexed family $v = (v_i)_{i \in J}$ of vectors $v_i \in \bbF_2^m$, let its \emph{independence index}, $\ind(v)$, denote the maximal number $k \leq \lvert J \rvert$ such that every subset of $k$ distinct indices selects linearly independent vectors. For a matrix $B \in \bbF_2^{m \times n}$, let $\ind(B)$ denote the independence index of its rows.

By the previous section, to find families of matrices with high CNOT-complexity, we can therefore aim to find families of matrices $U$ with high independence index, whose row count $a$ is asymptotically larger than the column count $b$, and then consider the family of matrices $M_0$ as above.

In \cite{Sergeev2025}, it is furthermore noted that matrices with high independence index can be obtained from the theory of linear codes, and that the BCH code in particular provides matrices with high additive complexity. The same construction works for our purposes, so we first recall the details. Let $p > 0$, and let $\alpha_1, \dots, \alpha_s$ be pairwise distinct nonzero elements of $\bbF_{2^p}$. The matrix
\[
  \begin{pmatrix}
    \alpha_1^1 & \alpha_1^2 & \cdots & \alpha_1^s \\
    \alpha_2^1 & \alpha_2^2 & \cdots & \alpha_2^s \\
    \vdots & \vdots & \ddots & \vdots \\
    \alpha_s^1 & \alpha_s^2 & \cdots & \alpha_s^s
  \end{pmatrix}
\]
has full rank over $\bbF_{2^p}$. Fix a basis of $\bbF_{2^p}$ as an extension over $\bbF_2$. The matrix $V$ obtained by expanding each entry of the above matrix in this basis then has full row rank $s$ over $\bbF_2$.

For sufficiently large $n$, choose $p = \lceil \log_2 n \rceil$, $s = \lceil \sqrt{n} \rceil$ and $m = ps$, and let $U \in \bbF_2^{(n-m) \times m}$ be the matrix obtained by choosing and expanding $n - m$ different elements of $\bbF_{2^p}$ as above, such that each chosen $\alpha \in \bbF_{2^p}$ gives rise to a row $(\alpha, \alpha^2, \dots, \alpha^s)$, expanded in the chosen basis. Then by the above observation, $\ind(U) \geq s$. The hard family of \cite{Sergeev2025} is the following (in which we have permuted the rows to match the above construction).
\begin{thm}[{\cite[Cor. 1]{Sergeev2025}}]
  \label{thm:sergeev-main}
  The explicit family $A_U = \begin{psmallmatrix} U & 0 \\ 0 & \Ut \end{psmallmatrix} \in \bbF_2^{n \times n}$ has additive complexity
  \[
    \additive(A_U) \geq 5n - O(\sqrt{n} \log n).
  \]
\end{thm}
\begin{thm}
  There is an unbounded set $I\subseteq\mathbb N$ and an explicit family $\{M_N\}_{N\in I}$, with $M_N\in\mathrm{GL}_N(\bbF_2)$ and $d(M_N)\geq4N-o(N)$.
\end{thm}
\begin{proof}
  Let $A_U = \begin{psmallmatrix} U & 0 \\ 0 & \Ut \end{psmallmatrix}$ be the family of matrices from Theorem~\ref{thm:sergeev-main} and let $M_0 = \begin{pmatrix} U & I_a \\ I_b & \Ut \end{pmatrix}$, where $a = n - m$, $b = m$. By Lemma~\ref{lem:add-permutation} and Theorem~\ref{thm:sergeev-main}, we have
  \[
    \additive(M_0) \geq \additive(A_U) - n \geq 4n - O(\sqrt{n} \log n).
  \]
  Now $M_0$ need not be invertible a priori, but by Lemma~\ref{lem:permutation-matrix}, we have $\mathrm{corank}(M_0) \leq b = O(\sqrt{n} \log n)$, so as in Lemma~\ref{lem:embedding}, we can embed $M_0$ as a submatrix of an invertible matrix $M_N \in \mathrm{GL}_N(\bbF_2)$, so as in \eqref{eq:embedding-bound}, we get
  \[
    d(M_N) \geq 4N - o(N).
  \]
  Finally, note that every matrix $M_N$, $N\in I$, is explicit, since each step of the construction may be performed in time polynomial in $N$.
\end{proof}

\section{An explicit hard instance}
In fact, we may trace through the construction to find a single explicit hard matrix; in the above, $s$ and $p$ are chosen to simplify the asymptotic analysis, but by appealing to Theorem~\ref{thm:sergeev-workhorse} directly, we can search for low values of $p$ and $s$ that give a matrix with high additive complexity.

To this end, we will build on the main technical result of \cite{Sergeev2025} that allows us to determine lower bounds for additive complexity. To do this, we need to introduce some notation. For $U \in \bbF_2^{a \times m}$, let $\reduced(U)$ denote its \emph{reduced complexity} \cite[§2]{Sergeev2025}. As above, let $A_U = \begin{psmallmatrix} U & 0 \\ 0 & \Ut \end{psmallmatrix} \in \bbF_2^{n \times n}$, $n = a + m$. Then by \cite[Lem.~1]{Sergeev2025} and \cite[Prop.~1]{Sergeev2025}, and the fact that $\reduced(U) \leq \additive(U)$, we have
\begin{align}
  \label{eq:additive-bound}
    \additive(A_U) \geq \additive(\Ut) + \reduced(U) = \additive(U) + (a - m) + \reduced(U) \geq 2 \reduced(U) + (a - m).
\end{align}
Let the \emph{weight} of a row be the number of nonzero entries in that row, and recall the following technical result:
\begin{thm}[{\cite[Thm. 2]{Sergeev2025}}]
  \label{thm:sergeev-workhorse}
  Let $m \leq a$, and let $B \in \bbF_2^{a \times m}$ have no rows of weight $1$ and satisfy $\ind(B) \geq 2k+2\geq 6$. Then
  \[
    \reduced(B) \geq a + \frac{2k-2}{2k+1} a^{k/(k+1)} - m.
  \]
\end{thm}
Now, return to the construction of the previous section. Let $p = 14$, and for our construction of $U$, choose all $a = 2^p - 1 = 16383$ nonzero elements of $\bbF_{2^p}$, let $s = 56$, and $m = sp = 784$, expanding in the polynomial basis $\{1, x, \dots, x^{13}\}$ of $\bbF_{2^{14}} = \bbF_2[x] / (x^{14} + x^5 + x^3 + x + 1)$. Construct $U \in \bbF_2^{a \times m}$ as above, so $\ind(U) \geq 56$.

\begin{prop}
  \label{prop:explicit-hard-instance}
  Let $U$ be the matrix just defined and let
  \[
    M_0=\begin{pmatrix}U&I_a\\I_m&\Ut\end{pmatrix}\in\bbF_2^{17167\times17167}.
  \]
  Then $M_0\in\mathrm{GL}_{17167}(\bbF_2)$ and
  \[
    d(M_0)\geq51535>3(17167-1)=51498.
  \]
\end{prop}
\begin{proof}
Each row of $U$ has weight at least $s$, since each of its $s$ nonzero field-element blocks contributes at least one nonzero binary coordinate. Thus, by~\eqref{eq:additive-bound} and Theorem~\ref{thm:sergeev-workhorse}, applied with $k = 27$, we have
\[
    \additive(A_U) \geq 2 \reduced(U) + (a - m) \geq 2\left(a + \frac{2k-2}{2k+1}a^{k/(k+1)} - m\right) + (a - m) \geq 68702.
\]
By Lemma~\ref{lem:add-permutation},
\[
    \additive(M_0) \geq \additive(A_U) - n \geq 51535.
\]
An exact computation over $\bbF_2$ gives $\mathrm{rank}(I_{784}+\Ut U)=784$. Lemma~\ref{lem:permutation-matrix} therefore shows that $M_0$ is invertible. Finally, $d(M_0)\geq\additive(M_0)\geq51535$, while $3(17167-1)=51498$.
\end{proof}
Thus $M_0$ requires strictly more CNOT gates to synthesize than the cyclic permutation on $17167$ symbols. The matrix itself is illustrated in Figure~\ref{fig:hard-instance}.
\begin{figure}[ht]
    \centering
    \includegraphics[width=0.6\textwidth]{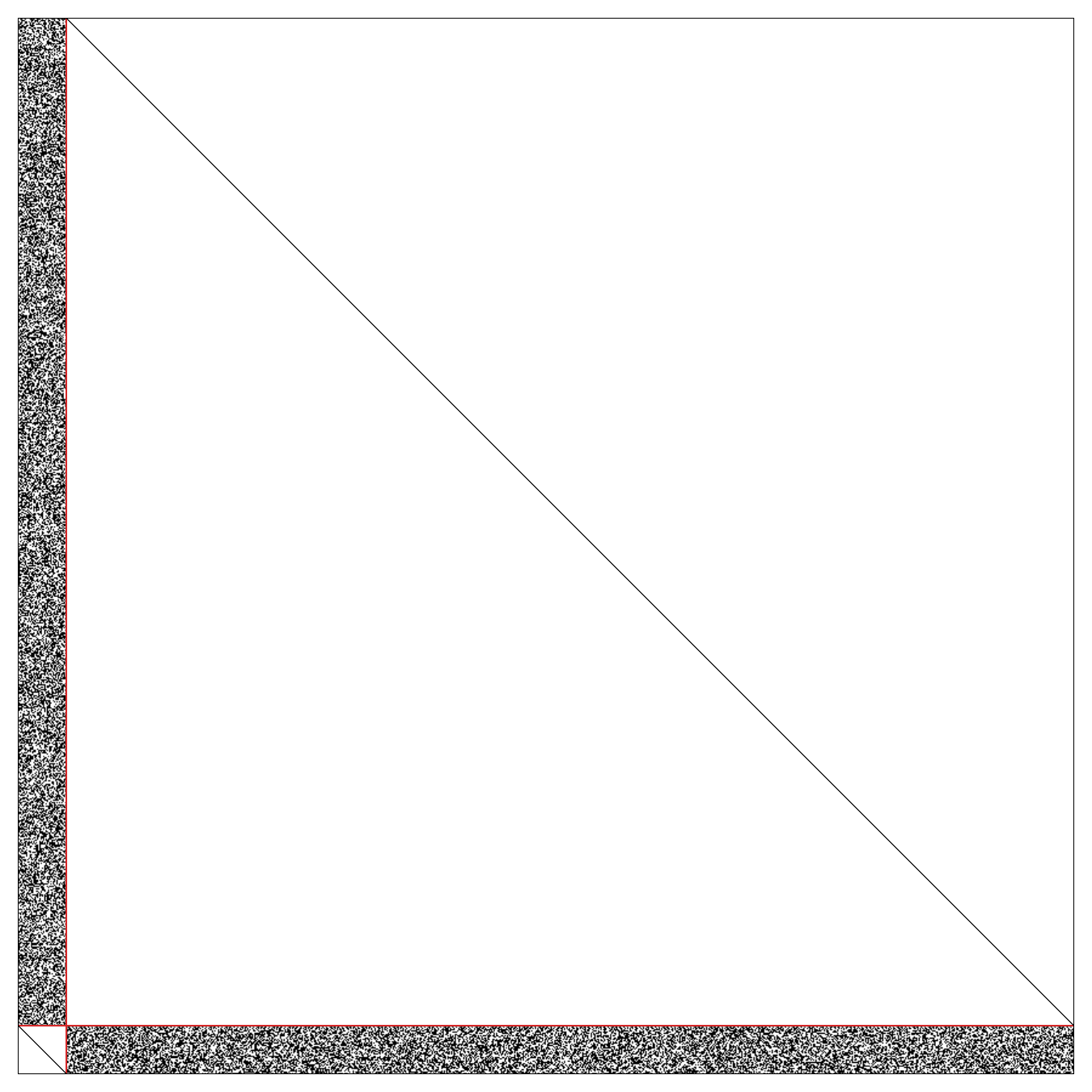}
    \caption{The explicit hard instance $M_0 \in \GLn$, $n = 17167$, for which we have $d(M_0) \geq 51535 > 3(n-1) = 51498$. The matrix is constructed from the parity-check matrix of the BCH code with $p = 14$ and $s = 56$.}
    \label{fig:hard-instance}
\end{figure}
\begin{rem}
  Recall that it is known \cite{Bataille2022,Christensen2025} that there is a hard instance in $\GLn$ for some $n$ with $8 \leq n \leq 20$, whereas the instance provided above has $n = 17167$. Similarly, the lower bound $d(M_0) \geq 51535$ is likely far from tight; running the synthesis algorithm of \cite{Patel2008}, splitting $M_0$ into 8 sections, yields a circuit with $40{,}015{,}039$ CNOT gates, so $51{,}535 \leq d(M_0) \leq 40{,}015{,}039$.
\end{rem}
\begin{rem}
  In the above, we have focused on the model of CNOT-complexity where the only allowed operation is the CNOT gate. Another common model of practical interest is that in which swapping two qubits is considered a free operation; in other words, given a matrix $A \in \GLn$, the goal is still to reduce $A$ to the identity using as few row additions as possible, but we are allowed to permute rows at no cost. Since permutations can always be collected at one end of a given circuit, without changing the number of row additions, the relevant complexity measure is the minimum
  \[
    d_{\mathrm{SWAP}}(A) = \min \{ d(PA) \mid P \text{ a permutation matrix} \}.
  \]
  As swapping is considered a free operation in the case of additive complexity, we have $\additive(A) = \additive(PA)$ for any permutation matrix $P$, and therefore
  \[
    d(A) \geq d_{\mathrm{SWAP}}(A) \geq \additive(A)
  \]
  for any $A \in \GLn$. The lower bounds established for explicit families in this work rely only on lower bounds for additive complexity and therefore hold in this model as well.
\end{rem}

The motivation for using parity-check matrices in the construction is the general fact that any linear code with minimum distance $d$ and parity-check matrix $H$ has $d = \ind(H^{\!\top}) + 1$, so the construction should be seen as suggesting that higher distance leads to harder instances.

In the context of quantum error correction, circuits are synthesized for encoding and syndrome extraction, and having to synthesize a circuit with a large number of entangling gates is undesirable, as it increases the overall error rate. A natural extension of the present work is to investigate whether similar arguments may be used to achieve lower bounds on circuit sizes for quantum error correcting codes obtained from classical linear codes, such as CSS codes.

\subsection*{Code availability}
The code for verifying the computational claims above is available at\\\url{https://github.com/kvantify/paper-lower-bounds-cnot}.

\bibliographystyle{is-alpha}
\bibliography{references}

@article{Amy2013,
  author    = {Matthew Amy and
               Dmitri Maslov and
               Michele Mosca and
               Martin Roetteler},
  title     = {{A Meet-in-the-Middle Algorithm for Fast Synthesis of Depth-Optimal
               Quantum Circuits}},
  journal   = {{IEEE} Trans. Comput. Aided Des. Integr. Circuits Syst.},
  volume    = {32},
  number    = {6},
  pages     = {818--830},
  year      = {2013},
  url       = {https://doi.org/10.1109/TCAD.2013.2244643},
  doi       = {10.1109/TCAD.2013.2244643},
  bibsource = {dblp computer science bibliography, https://dblp.org}
}

@article{Amy2018,
  doi = {10.1088/2058-9565/aad8ca},
  url = {https://doi.org/10.1088/2058-9565/aad8ca},
  year = {2018},
  publisher = {IOP Publishing},
  volume = {4},
  number = {1},
  pages = {015002},
  author = {Amy, Matthew and Azimzadeh, Parsiad and Mosca, Michele},
  title = {{On the controlled-NOT complexity of controlled-NOT--phase circuits}},
  journal = {Quantum Science and Technology}
}

@article{Bataille2022,
    title = {{Quantum Circuits of {CNOT} gates: Optimization and Entanglement}},
    shorttitle = {Quantum Circuits of \$\${\textbackslash}texttt \{\vphantom\}{{CNOT}}\vphantom\{\}\$\$gates},
    author = {Bataille, Marc},
    year = {2022},
    month = jul,
    journal = {Quantum Information Processing},
    volume = {21},
    number = {7},
    pages = {269},
    issn = {1573-1332},
    doi = {10.1007/s11128-022-03577-8},
    url = {https://doi.org/10.1007/s11128-022-03577-8},
}

@inproceedings{Brugiere2020,
  author       = {Timoth{\'{e}}e Goubault de Brugi{\`{e}}re and Marc Baboulin and Beno{\^{\i}}t Valiron and Simon Martiel and Cyril Allouche},
  title        = {{Quantum {CNOT} Circuits Synthesis for {NISQ} Architectures Using the Syndrome Decoding Problem}},
  booktitle    = {Reversible Computation - 12th International Conference, {RC} 2020,
                  Oslo, Norway, July 9-10, 2020, Proceedings},
  series       = {Lecture Notes in Computer Science},
  volume       = {12227},
  pages        = {189--205},
  publisher    = {Springer},
  year         = {2020},
  url        = {https://doi.org/10.1007/978-3-030-52482-1\_11},
  doi          = {10.1007/978-3-030-52482-1\_11},
  bibsource    = {dblp computer science bibliography, https://dblp.org}
}

@article{Brugiere2021,
    author = {De Brugi\`{e}re, Timoth\'{e}e Goubault and Baboulin, Marc and Valiron, Beno\^{\i}t and Martiel, Simon and Allouche, Cyril},
    title = {{Gaussian Elimination versus Greedy Methods for the Synthesis of Linear Reversible Circuits}},
    year = {2021},
    issue_date = {September 2021},
    publisher = {Association for Computing Machinery},
    address = {New York, NY, USA},
    volume = {2},
    number = {3},
    url = {https://doi.org/10.1145/3474226},
    doi = {10.1145/3474226},
    journal = {ACM Transactions on Quantum Computing},
    month = {sep},
    articleno = {11},
    numpages = {26}
}

@ARTICLE{Bu2025,
  title     = "Minimum synthesis cost of {CNOT} circuits",
  author    = "Bu, Alan and Fan, Evan and Joo, Robert",
  journal   = "Quantum Information Processing",
  publisher = "Springer Science and Business Media LLC",
  volume    =  24,
  number    =  7,
  month     =  jul,
  year      =  2025,
  copyright = "https://www.springernature.com/gp/researchers/text-and-data-mining",
  language  = "en",
  doi       = "10.1007/s11128-025-04831-5",
  url       = "https://doi.org/10.1007/s11128-025-04831-5"
}

@INCOLLECTION{Christensen2025,
  title     = "On exact sizes of minimal {CNOT} circuits",
  booktitle = "Lecture Notes in Computer Science",
  author    = "Christensen, Jens Emil and J{\o}rgensen, S{\o}ren Fuglede and
               Pavlogiannis, Andreas and van de Pol, Jaco",
  publisher = "Springer Nature Switzerland",
  pages     = "71--88",
  series    = "Lecture Notes in Computer Science",
  year      =  2025,
  address   = "Cham",
  copyright = "https://www.springernature.com/gp/researchers/text-and-data-mining",
  language  = "en",
  doi       = "10.1007/978-3-031-97063-4\_6",
  url       = "https://doi.org/10.1007/978-3-031-97063-4\_6"
}

@ARTICLE{Gheorghiu2023,
  author={Gheorghiu, Vlad and Huang, Jiaxin and Li, Sarah Meng and Mosca, Michele and Mukhopadhyay, Priyanka},
  journal={IEEE Transactions on Computer-Aided Design of Integrated Circuits and Systems}, 
  title={{Reducing the {CNOT} Count for {Clifford+T} Circuits on {NISQ} Architectures}}, 
  year={2023},
  volume={42},
  number={6},
  pages={1873-1884},
  doi={10.1109/TCAD.2022.3213210},
  url={https://doi.org/10.1109/TCAD.2022.3213210}
}

@MISC {Gowers2011,
    TITLE = {What is the complexity of this problem?},
    YEAR = {2011},
    AUTHOR = {Timothy Gowers},
    HOWPUBLISHED = {MathOverflow},
    EPRINT = {https://mathoverflow.net/q/69873},
    URL = {https://mathoverflow.net/q/69873}
}

@ARTICLE{Jukna2013,
  title     = "Complexity of linear Boolean operators",
  author    = "Jukna, Stasys and Sergeev, Igor",
  journal   = "Found. Trends Theor. Comput. Sci.",
  publisher = "Emerald",
  volume    =  9,
  number    =  1,
  pages     = "1--123",
  month     =  oct,
  year      =  2013,
  language  = "en",
  isbn = " 1601987269",
  doi = "10.1561/0400000063",
  url = "https://doi.org/10.1561/0400000063"
}

@inproceedings{Kazuo2002,
author = {Iwama, Kazuo and Kambayashi, Yahiko and Yamashita, Shigeru},
title = {{Transformation rules for designing CNOT-based quantum circuits}},
year = {2002},
isbn = {1581134614},
publisher = {Association for Computing Machinery},
address = {New York, NY, USA},
url = {https://doi.org/10.1145/513918.514026},
doi = {10.1145/513918.514026},
booktitle = {Proceedings of the 39th Annual Design Automation Conference},
pages = {419--424},
numpages = {6},
location = {New Orleans, Louisiana, USA},
series = {DAC '02}
}

@inproceedings{Meuli2018,
  author       = {Giulia Meuli and Mathias Soeken and Giovanni De Micheli},
  title        = {{{SAT}-based \{CNOT, T\} Quantum Circuit Synthesis}},
  booktitle    = {Reversible Computation - 10th International Conference, {RC} 2018,
                  Leicester, UK, September 12-14, 2018, Proceedings},
  series       = {Lecture Notes in Computer Science},
  volume       = {11106},
  pages        = {175--188},
  publisher    = {Springer},
  year         = {2018},
  url          = {https://doi.org/10.1007/978-3-319-99498-7\_12},
  doi          = {10.1007/978-3-319-99498-7\_12},
  bibsource    = {dblp computer science bibliography, https://dblp.org}
}

@article{Patel2008,
    author = {Patel, Ketan N. and Markov, Igor L. and Hayes, John P.},
    title = {Optimal synthesis of linear reversible circuits},
    year = {2008},
    issue_date = {March 2008},
    publisher = {Rinton Press, Incorporated},
    address = {Paramus, NJ},
    volume = {8},
    number = {3},
    issn = {1533-7146},
    journal = {Quantum Info. Comput.},
    pages = {282--294},
    numpages = {13},
}

@article{Schneider2022,
  title={{A {SAT} Encoding for Optimal {Clifford} Circuit Synthesis}},
  author={Sarah Schneider and Lukas Burgholzer and Robert Wille},
  journal={2023 28th Asia and South Pacific Design Automation Conference (ASP-DAC)},
  year={2022},
  pages={190-195},
  doi = {10.1145/3566097.3567929},
  url={https://doi.org/10.1145/3566097.3567929}
}

@ARTICLE{Sergeev2025,
  title     = "{Lower Bounds for Additive Complexity of Linear Operators and Bilinear Algorithms for Matrix and Polynomial Multiplication {$GF(2)$}}",
  author    = "Sergeev, Igor S.",
  journal   = "Math. Notes",
  publisher = "Pleiades Publishing Ltd",
  volume    =  118,
  number    = "3-4",
  pages     = "848--862",
  month     =  oct,
  year      =  2025,
  copyright = "https://www.springernature.com/gp/researchers/text-and-data-mining",
  language  = "en",
  doi = {10.1134/S0001434625605143},
  url = {https://doi.org/10.1134/S0001434625605143}
}

@inproceedings{Shaik2024,
  author       = {Irfansha Shaik and
                  Jaco van de Pol},
  title        = {{Optimal Layout-Aware {CNOT} Circuit Synthesis with Qubit Permutation}},
  booktitle    = {{ECAI}},
  series       = {Frontiers in Artificial Intelligence and Applications},
  volume       = {392},
  pages        = {4207--4215},
  publisher    = {{IOS} Press},
  year         = {2024},
  doi = {10.3233/FAIA240993},
  url = {https://doi.org/10.3233/FAIA240993}
}

@MISC{Webster2025,
  title         = "Heuristic and optimal synthesis of {CNOT} and {Clifford}
                   circuits",
  author        = "Webster, Mark and Koutsioumpas, Stergios and Browne, Dan E",
  month         =  mar,
  year          =  2025,
  copyright     = "http://creativecommons.org/licenses/by/4.0/",
  archivePrefix = "arXiv",
  primaryClass  = "quant-ph",
  eprint        = "2503.14660",
  url = {https://doi.org/10.48550/arXiv.2503.14660},
}

\end{document}